\documentclass{article}

\usepackage{epstopdf}

\usepackage{amsfonts}
\usepackage{amsmath}
\usepackage{amsthm}
\usepackage{float,graphicx}
\usepackage{subfigure}

\usepackage{bbm}
\usepackage{natbib}
\usepackage{cases}
\usepackage[colorlinks,             linkcolor=blue, anchorcolor=blue,
citecolor=blue,             ]{hyperref}%
\theoremstyle{plain}
\newtheorem{theorem}{Theorem}[section]

\theoremstyle{definition}
\newtheorem{definition}{Definition}

\theoremstyle{remark}
\newtheorem{remark}{Remark}

\begin{document}


\title{Constructing long-short stock portfolio with \\a new listwise learn-to-rank algorithm}

\author{Xin Zhang
\thanks{Email: alpoise@pku.edu.cn},~Lan Wu\thanks{Corresponding author. Email: lwu@pku.edu.cn}\\
\vspace{1ex}
School of Mathematical Sciences, Peking University, China\\
Zhixue Chen\\
IIIS, Tsinghua University, China}
\date{Dec 7,2020}
\maketitle
\begin{abstract}
Factor strategies have gained growing popularity in industry with the fast development of machine learning. Usually, multi-factors are fed to an algorithm for some cross-sectional return predictions, which are further used to construct a long-short portfolio. Instead of predicting the value of the stock return, emerging studies predict a ranked stock list using the mature learn-to-rank technology. In this study, we propose a new listwise learn-to-rank loss function which aims to emphasize both the top and the bottom of a rank list. Our loss function, motivated by the long-short strategy, is endogenously shift-invariant and can be viewed as a direct generalization of ListMLE. Under different transformation functions, our loss can lead to consistency with binary classification loss or permutation level 0-1 loss. A probabilistic explanation for our model is also given as a generalized Plackett-Luce model. Based on a dataset of 68 factors in China A-share market from 2006 to 2019, our empirical study has demonstrated the strength of our method which achieves an out-of-sample annual return of 38\% with the Sharpe ratio being 2.    

\end{abstract}
\vspace{2ex}

\newpage

\section{Introduction}
Ever since the birth of active management, which can date back to as early as the publication of the CAPM \citep{sharpe1964capital,lintner1965security} and the Fama-French 3 factor model \citep{fama1993common}, there have been numerous well-documented pros and cons of its value. While \citet{carhart1997persistence} stands as a capstone of the conventional wisdom that ``the results do not support the existence of skilled or informed mutual fund portfolio managers'', Cremer's review of the most recent literature \citep{cremers2019challenging} actually challenges this opinion and suggests that the conventional wisdom is too negative, in the sense of both market-timing and stock-picking. For market-timing, this involves a prediction towards the market movement whose period ranges from months to milliseconds. As for stock-picking strategy, a basic framework is to find factors and combine them to generate a prediction for next period's returns. To gain this cross-sectional return, one can adopt long-short factor strategies that long the top ranked equities and short the bottom ranked ones, as is inspired by Fama and French. While a highly predictive factor could be confidential and printing money somewhere in this world, hundreds of trivial factors have been put forward in literatures \citep{tulchinsky2019finding,giglio2019thousands}. An open question is whether it is possible to generate strong factors from the trivial ones. We try to answer this question with a novel learn-to-rank method which, motivated by the nature of long-short trading behavior, aims to directly predict a cross-sectional ranked stock list rather than some exact stock returns.

Learn-to-rank is a class of supervised machine learning algorithms that has proven to be extremely successful in information retrieval (IR) field. Ranking algorithm has also become the core part of the recommender system that has been broadly used in web search, newsfeed, online shopping, and advertisement. Machine learning algorithms based on ranking has also demonstrated its power in text summarization and machine translation. However, these ranking algorithms have not gained enough attention in factor investing yet. Although many other machine learning algorithms have been used for developing factor strategies \citep{de2018advances,rasekhschaffe2019machine}, they do not view this problem from a ranking perspective. A crucial difference lies in that for IR we only care about the accuracy at the top, but for long-short strategy we want both the top and the bottom to be accurate. To bridge this gap, we need a learn-to-rank method that emphasizes both the top and the bottom. As a matter of fact, some researchers have casted the factor strategy into the ranking framework, but they mostly focus on adding alternative factors or neural network components \citep{song2017stock,feng2019temporal,fang2020neural}. As for satisfying the long-short need, very few have been covered. In \cite{song2017stock}, the authors reverse the rank labels and fit the model twice, in the hope that the two models can predict the top and bottom respectively. But despite the troublesome parameter tuning stuff, this approach may have contradictory outputs from the two models. Hence it may not guarantee a coherent rank and could tell us to buy and sell one stock at the same time. This motivates us to propose a new learn-to-rank algorithm that targets at a coherent rank list that views the top and the bottom equally important.        

In this paper, instead of the absolute stock returns, we focus on predicting the relative rank of the returns. This preference is explained not only by portfolio manager's task to beat a relative index, but also by the difficulty of making value prediction. This difficulty mainly arises from the blurry boundary of the input information and the low information-noise ratio of financial data. For example, some unexpected overnight news could severely affect the whole market; some stock prices could be generated by unpredictable trading behaviors. It is really not so clear any experts or factors can reliably predict individual stock returns, but somehow we are more comfortable to anticipate stock moves relative to other stocks or to a factor model. This rank opinion is also shared by \cite{song2017stock,feng2019temporal,zhu2011hybrid,wang2018stock} and those who use rank information coefficient (IC) to evaluate their factors.

Our main contribution of this study is twofold. First, we develop a new category of learn-to-rank loss function for long-short factor strategy. To the best of our knowledge, we are the first to propose a top and bottom focused learn-to-rank framework for long-short strategy. Particularly, we design a novel type of surrogate loss function and discuss its theoretical properties. Our loss function is shift-invariant due to its symmetric nature, and it can induce surrogate loss that is consistent with binary classification loss or permutation level 0-1 loss under different transformation functions. We also give a probabilistic explanation to the model. Secondly, we conduct a detailed empirical study to examine the performance of our model in China A-share market. We've achieved an annualized return of 38\% with the Sharpe Ratio being 2. Our method outperforms MLP, ListMLE and Song's using ListMLE twice.

The rest of this paper is organized as follows. Section \ref{backgrounds} gives a brief introduction to the background knowledge, including the factor strategy, learn-to-rank framework and ListMLE algorithm. Section \ref{our model} presents our model and its theoretical analysis. Section \ref{empirical} compares our model and some other models empirically in China A-share market. And we summarize our key findings and some ideas for future study in Section \ref{conclusion}.

\section{Backgrounds}
\label{backgrounds}

\subsection{Factor Strategy}
Factors are at the core of factor strategies. Hundreds of potential pricing factors have been published in literatures \citep{harvey2016and,hou2017comparison}, even more could be found in industry. While one branch of the factor strategy is to test these factors in a multiple testing framework \citep{feng2020taming}, another branch, which is also more popular among practitioners, is to generate stronger factors from the factor zoo. The most traditional ways of combing factors might be sequential filtering, majority vote and linear regression. Econometricians and statisticians further extend the linear model from different perspectives, such as SVM, Adaboost and graphical model \citep{liu2016semiparametric}. Machine learning has also been used both in generating and combing factors. For example, with textual analysis, machine learning is able to extract sentimental factors from the news. As for combing factors, these algorithms usually view the factors as the input and the returns as the output, and cast the problem into a classification or regression problem.
 
Long-short strategy is a natural generalization of pure long investing when investors want to capitalize an overvalued asset. \citet{jacobs1993long} have examined the ways of implementing long-short strategies, the theoretical and practical benefits, and some practical concerns of long-short strategies. Nowadays, long-short strategy has become extremely popular among hedge funds. A popular way of long-short factor investing is to first sort the stocks, divide them into 10 groups, then long the top ranked group and short the bottom ranked group. This ranking approach is welcomed since it does not require estimating the factor loadings. 
  
\subsection{Overview of Learn-to-rank}
Learn-to-Rank, originating from webpage search, can include many of the previous methods as its subset. Its basic framework is to take documents (in our scenario, stocks) along with their features as the input and a rank list based on the corresponding relevance judgments as the output. A surrogate loss function is employed to learn a scoring function which assigns a score to each document based on their features such that as long as we sort the documents based on their scores we will give a good prediction of their rank. The scoring function does not vary from documents. Most learn-to-rank algorithms can be categorized into pointwise, pairwise and listwise approach based on their loss function. 

To evaluate the prediction, especially when we care more about whether the top positions have been ranked correctly, NDCG is the most popular metric which is defined as follows:
	\begin{equation}
	\text{NDCG}@k(\pi,l) = \frac{1}{Z_k}\sum_{j=1}^{k}G(l_{\pi^{-1}(j)})\eta(j)
	\end{equation}
	where $\pi$ is the predicted list, $\pi^{-1}(j)$ denotes the document ranked at position $j$ of the list $\pi$. $l$ stands for the relevance judgments and $G(z) = (2^z-1)$ is a usual rating function of the document. $\eta(j)$ is a position discount factor (usually set to be $1/\log(1+j)$). The cutoff of the position $k$ means that we only care about the accuracy of the first $k$ positions.  $Z_k$ is the normalizing value to set NDCG fall in the range $[0,1]$. So we can use 1-NDCG as the true loss. Besides, permutation level 0-1 loss reaches 0 if and only if the two lists are identical. Binary classification loss is to first label the top 50\% of the rank list as 1 and the rest -1. Then a rank list $\pi$ achieves 1 if and only if it labels every document correctly. Among many other ranking metrics, GAUC \citep{song2015recommending} also considers the accuracy of the top and the bottom of a rank list, but its limitation lies in that it only considers the special case of $\{1,0,-1\}$ labels.

\subsection{ListMLE}

ListMLE is a state-of-the-art listwise learn-to-rank algorithm, which defines the probability distribution based on the Plackett-Luce Model in a top-down style. It aims to utilize a likelihood loss as the surrogate loss, defined as:
	\begin{equation}
		\mathcal{L}(f,x,y) = -\log \mathbb{P}(y|x;f) = -\log \prod_{i=1}^{n}\frac{\psi(f(x_{y^{-1}(i)}))}
		{\sum_{k=i}^{n}\psi(f(x_{y^{-1}(k)}))},
		\label{loss_listmle}
	\end{equation}
	where $y^{-1}(i)$ represents the item that is labeled at the $i$-th position. $x$ is the feature vector, $n$ is the sample size, and $f$ is the scoring function. The function $\psi(\cdot)$ is the transformation function that maps the score to $\mathbb{R}^+$. $\psi(\cdot)$ is usually taken to be linear, exponential or sigmoid. For simplicity, we denote $\psi_i:=\psi(f(x_{y^{-1}(i)}))$.

A probabilistic explanation of the Plackett-Luce model is the vase model metaphor given by \citet{silverberg1980statistical}. Consider drawing balls from a vase full of colored balls. The number of balls of each color is in proportion to $\psi_i$. Suppose there are infinite number of balls if non-rational proportions are needed. At the first stage a ball $c_1$ is drawn from the vase; the probability of this selection is $\psi_1/\sum_{i=1}^{n}{\psi_i}$. At the second stage, another ball is drawn - if it is the same color as the first, then put it back, and keep on trying until a new color $c_2$ is selected; the probability of this second selection is $\psi_2/\sum_{i=2}^{n}{\psi_i}$. Continue through the stages until a ball of each color has been selected. Then the probability of the color sequence is as shown in equation \ref{loss_listmle}.


Previous theoretical analysis on ListMLE has shown that it is consistent with the permutation-level 0-1 loss. Intuitively speaking, this means that for $n$ fixed scores ${\psi_1,...,\psi_n}$, the loss defined in equation \ref{loss_listmle} achieves its minimal among all permutations if it is the descending sequence of ${\psi_1,...,\psi_n}$. A strict definition of consistency in \cite{xia2008listwise} is given as follows:
	\begin{definition}
		We define $\Lambda_y$ as the space of all possible probabilities on the permutation space Y, i.e, $\Lambda_y:=\{p\in R^{|Y|}:\sum_{y\in Y}p_y=1,p_y\geq 0\}$. 	
	\end{definition}

	\begin{definition}
	The loss $\phi_y(g)$ is consistent on a set $\Omega \subset R^n$ with respect to the permutation-level 0-1 loss, if the following conditions hold: $\forall p \in \Lambda_y$, assume $y^*=\arg \max_{y\in Y} p_y$ and $Y_{y^*}^c$ denotes the space  of permutations after removing $y^*$, we have
	$$
	\inf_{g\in \Omega}Q(g) < \inf_{g\in \Omega, \text{ sort}(g)\in Y_{y^*}^c} Q(g), \quad
	\text{where} \ Q(g) = \sum_{y\in Y}P(y|x)\phi_y(g(x)).
	$$
	\end{definition}
	
Although ListMLE is theoretically consistent with the permutation level 0-1 loss (and not the NDCG loss), it actually has good empirical performance measured by NDCG\citep{tax2015cross,xia2008listwise,gao2014democracy}. Many extensions of ListMLE might also be inspiring when casting ListMLE into factor strategy, such as position-aware ListMLE \citep{lan2014position}, dyad ranking \citep{schafer2015dyad} and multi-view ranking \citep{gao2014democracy}.

\section{Our Model}
\label{our model}
In this section, we will propose our model --- ListFold. Motivated by the long-short strategy, we suggest a new kind of surrogate loss function that views the top and the bottom equally important as they both contribute to the portfolio's pnl. Without loss of generality, we assume we are given even number of stocks (documents) to rank.

\subsection{ListFold}
For $2n$ documents $X_1,...,X_{2n}$, the observed rank $y$ and the scoring function $f$, we try to decompose a permutation into an ordered stepwise pair selection procedure: the first long-short pair, the second long-short pair until the $n$-th long-short pair. So we can define a probability as follows:
	\begin{equation}
		\mathbb{P}_c(y|X,f) = \prod_{i=1}^{n} \frac{\psi(f_i - f_{2n+1-i})}
		{\sum_{i\leq u \neq v\leq 2n+1-i}\psi(f_u-f_v)},\label{prob}
	\end{equation}
	where $f_i := f(X_{y^{-1}(i)})$ represents the score of the document observed at the $i$-th position and $\psi$ is the transformation function as in ListMLE. The loss function is then defined as the negative log-likelihood:
	\begin{align}
	\mathcal{L}_c(f,y,X) &= -\log\mathbb{P}_c(y|X,f) \notag\\
		&= -\sum_{i=1}^{n}
			\left(\log{\psi(f_i - f_{2n+1-i}})-\log\sum_{i\leq u \neq v\leq 2n+1-i}\psi(f_u-f_v)\right). \label{llk}
	\end{align}

The intuition behind our loss setting is similar with ListMLE which decomposes the permutation probability into a stepwise conditional probability. The difference lies in that for each step, instead of picking out one document, our goal is to pick out a \textit{pair} that has the maximal score difference and further place them in the correct pairwise preference order. Based on the previous $i-1$ steps being ranked correctly, we can write the conditional probability as:
	\begin{align}
		\mathbb{P}_c^i:&=P_i \left (y^{-1}(i,2n+1-i)\big{\vert} X,y^{-1}(1,2n),...,y^{-1}(i-1,2n-i+2);f\right ) \notag \\
		&=\frac{\psi(f_i - f_{2n+1-i})+\psi(f_{2n+1-i}-f_i)}
			{\sum_{i\leq u\neq v\leq 2n+1-i}\psi(f_u-f_v)}
			\cdot
			\frac{\psi(f_i - f_{2n+1-i})}
			{\psi(f_i - f_{2n+1-i})+\psi(f_{2n+1-i}-f_i)} \notag \\
		&= \frac{\psi(f_i - f_{2n+1-i})}{\sum_{i\leq u \neq v\leq 2n+1-i}\psi(f_u-f_v)},\quad i=1,...,n.
	\end{align}
This gives a natural explanation to the symmetry of the denominator in $\mathbb{P}_c$. Also, due to this symmetry, ListFold is born to be shift-invariant, while ListMLE is shift-invariant only when $\psi$ is exponential.

When $\psi$ is exponential, a probabilistic explanation of our loss can be derived as a natural generalization of the vase model. Consider a multi-stage experiment of throwing darts. There are $2n$ planks stacked together. The width, length and height are $(w_i,l_i,1)$ respectively, subject to the restriction that $w_i*l_i=1$. At the first stage, simultaneously, person A and person B each randomly throw a dart towards the width and length direction of the planks. If their darts fall on the same plank then put the plank back and re-throw. Otherwise, mark A's plank $A_1$ and B's plank $B_1$. Continue through the stages without putting back the marked planks and mark the planks $A_i, B_i$ at step $i$, until all the planks have been marked. Then the probability of the planks sequence $\{A_1,...,A_n,B_n,...,B_1\}$ is:
	\begin{equation}
		\mathbb{P}= \prod_{i=1}^{n}\mathbb{P}_i:= \prod_{i=1}^{n}\frac{w_{A_i}*l_{B_i}}{(\sum_{j=i}^{n}w_{A_j}+w_{B_j})*(\sum_{j=i}^{n}l_{A_j} + l_{B_j}) - 2(n+1-i)},
	\end{equation}
this is equation \ref{llk} taking $\psi$ exponential and $f_i = \log(w_i)$.

\begin{figure}[H]
	\centering
	\includegraphics[scale=0.4]{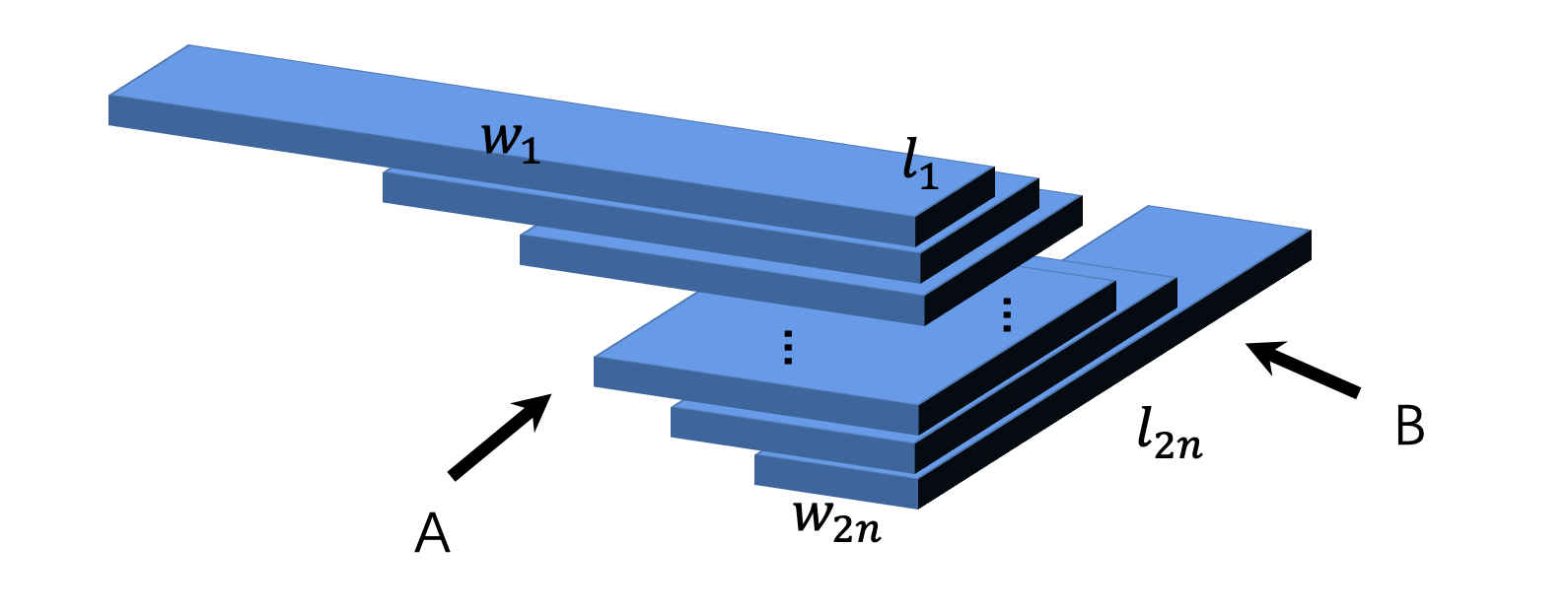}
	\label{plank}
	\caption{The probabilistic explanation of $\mathcal{L}_c^e$}	
\end{figure}

Note that, our loss is not just a trivial generalization of ListMLE. Actually, a naive generalization could be
	\begin{equation}
		\mathbb{P}_t(\pi|X,f) = 
		\prod_{i=1}^{n} \frac{\psi(f_i)}{\sum_{i\leq u\leq 2n}\psi(f_u)}
		\times 
		\prod_{j=1}^n \frac{\psi(-f_{2n+1-j})}
		{\sum_{j\leq u\leq 2n}\psi(-f_{2n+1-u})},
	\end{equation} 
which is more like combing two ListMLE by reversing the labels. The consistency of this surrogate loss with the permutation level 0-1 loss holds naturally (since the ground truth permutation achieves minimum on both parts, it would also minimize their product). However, $\mathbb{P}_t$ does not define a probability on the \textit{permutation space}, whereas our generalization $\mathbb{P}_c$ implies a probability model when $\psi$ is exponential. As we we will see in the following analysis, our loss also implies a pairwise intuition compared with $\mathbb{P}_t$.  
 
\subsection{Theoretical Analysis}

 In this part, we discuss the consistency of our loss function $\mathcal{L}_c$ with respect to the permutation level 0-1 loss. We denote $\mathcal{L}_c$ as $\mathcal{L}_c^s$ and $\mathcal{L}_c^e$ respectively when $\psi$ is sigmoid or exponential.

Before we start, we first spare some effort in understanding why this is a new challenging problem. In the previous work of ListMLE, the surrogate loss functions are all order sensitive, which basically says that if we exchanged the positions of any two documents towards the ground truth, the loss would decrease. The formal definition of order sensitive can be found in \cite{xia2008listwise}. Order sensitiveness means that as long as the neural network learns $f$ better at any two points, the loss will reduce. But we might actually prefer $f$ having a more holistic view when the loss reduces --- if we admit in the stock market we won't predict the rank perfectly correct anyway, we should be willing to allow the surrogate loss explores something more complex rather than just being order sensitive and allows no local minimum. From this perspective, our loss function $\mathcal{L}_c^e$ serves as an anomaly that might trigger new idea for proposing surrogate loss functions. Consider permutation over four numbers (5,4,1,0). Suppose we start with $\mathcal{L}_c^e(1,5,4,0) = 4.78$, and if we exchange the first two numbers, the loss actually increases: $\mathcal{L}_c^e(5,1,4,0) = 6.65$, whereas the ground truth $\mathcal{L}_c^e(5,4,1,0) = 0.65$ is still the minimal.

Now we are going to state two theorems that characterize the consistency of $\mathcal{L}_c^s$ and $\mathcal{L}_c^e$ respectively. These two theorems try to answer what kind of true loss we are targeting at when we minimize ListFold defined in equation \ref{llk}. The theoretical analysis will help us interpret the model outputs better and construct optimal long-short portfolio accordingly.

\begin{theorem}
	If the transformation function $\psi$ is sigmoid, suppose $a_1 \geq a_2 \geq...\geq a_n \geq b_n \geq b_{n-1} \geq ... \geq b_1$, and $\bold{f}:=(f_1,...,f_{2n})$ is a permutation of all $a_i$'s and $b_i$'s. Then our loss $\mathcal{L}_c^s(f,y,X)$ is consistent with the binary classification loss.
\end{theorem}

\begin{proof} \
 
	By definition $\text{sigmoid}(x) = 1/(1+e^{-x})$ and the property that $\text{sigmoid}(x) + \text{sigmoid}(-x) = 1$,
	\begin{align}
	\mathcal{L}_c^s(f,y,X) 
		&= -\sum_{i=1}^{n}
			\left(\log{\psi(f_i - f_{2n+1-i}})-\log\sum_{i\leq u \neq v\leq n+1-i}\psi(f_u-f_v)
			\right) \notag\\
	&= -\sum_{i=1}^{n}
			\left(\log{\text{sigmoid}(f_i - f_{2n+1-i}}) + (2n-2i+2)\right) \notag \\
			&= \sum_{i=1}^{n}
			\log (1+e^{-f_i + f_{2n+1-i}}) + C_n,
	\end{align}
	where $C_n=n(n+1)$ is a constant for a fixed $n$. Since $\log(1+e^{-x})$ is convex and monotone increasing, for any four scores $f_i \leq f_j \leq f_k \leq f_l$, consider all their permutations $(a,b,c,d) \in \text{Perm}(f_i,f_j,f_k,f_l)$ we make two pairs out of them:$(a,b),(c,d)$ and consider the loss:
	$$
	\mathcal{L}_c^s(a,b,c,d) = \log(1+e^{-(a-b)}) + \log(1+e^{-(c-d)}).
	$$	
	To minimize the loss, it suffices $a>b,c>d$ otherwise we can always swap their positions to reduce the loss. So we only need to compare three cases $f_l-f_i,f_l-f_j,f_l-f_k$ :
	\begin{align}
	& \log(1+e^{-(f_l-f_j)}) + \log(1+e^{-(f_k-f_i)}) \leq \log(1+e^{-(f_l-f_i)}) + \log(1+e^{-(f_k-f_j)}), \notag \\
	& \log(1+e^{-(f_l-f_j)}) + \log(1+e^{-(f_k-f_i)}) \leq \log(1+e^{-(f_l-f_k)}) + \log(1+e^{-(f_j-f_i)}). \notag
	\end{align}
	The first inequality is because of the convexity and the second one is because of the monotonicity.
	Keep using this rule for every two pairs, it is straightforward that the loss is minimized as long as the permutation pairs together $(a_1,b_n),(a_2,b_{n-1}),...,(a_n,b_1)$. The permutation of these $n$ pairs actually makes no difference. So $\mathcal{L}_c^s$ is consistent with the binary classification loss but not the permutation level 0-1 loss. 
\end{proof}

\begin{theorem}
	If the transformation function $\psi$ is exponential, suppose $a_1 \geq a_2 \geq...\geq a_n \geq b_n \geq b_{n-1} \geq ... \geq b_1$, and $(f_1,...f_n)$ is a permutation of all $a_i$'s and $(f_{-1},...f_{-n})$ is a permutation of all $b_i$'s, then denote $\bold{f}:=(f_1,...f_n,f_{-n},...,f_{-1})$. Then our loss $\mathcal{L}_c^e(f,y,X)$		
	$$
	\mathcal{L}_c^e(\bold{f},y,X) = \sum_{i=1}^{n}\left ( -(f_i - f_{-i})+\log \sum_{-i\leq s\neq t \leq i}e^{f_s-f_t}\right)
	$$
achieves its minimum at the descending sequence $\bold{f}^*=(a_1,...,a_n,b_n,...,b_1)$.
\label{thm2}
\end{theorem}

\begin{proof} \
 
	If $f_i < f_{-i}$, we can always reduce the loss by exchanging their positions, so the result holds naturally for n = 1. Suppose the result holds for 1,...,n, we next prove the result for n+1.
	
	Consider $2n+2$ numbers $\{ \alpha, a_1,...,a_n,b_n,...,b_1,\beta\}$. Denote $S_k = \{a_k,...,a_n,b_n,...,b_k\}$, and $S_k^\alpha = S_k - a_k + \alpha$, where ``$+$'' and ``$-$'' are taken over sets, representing adding and removing an element from a set. Denote the permutation space over the set $S_k$ by $\text{Perm}(S_k)$ and denote $\ell(S_k):= \log \sum_{-k\leq s\neq t \leq k}e^{f_s-f_t}$. Then for any permutation $(f_1,...,f_n, f_{-n},...,f_{-1}) \in \text{Perm}(S_1)$, by the following decomposition: 
	\begin{equation}
		\mathcal{L}_c^e([\alpha, f_1,...,f_n,f_{-n},...,f_{-1},\beta])
		= -(\alpha-\beta)+ \ell (S_1+\alpha+\beta)+ \notag \\
		\mathcal{L}_c^e([f_1,...,f_n,f_{-n},...,f_{-1}]),
	\end{equation}
	and the observation that the permutation of $\{f_1,...,f_{-1}\}$ doesn't influence the first two items, the following holds immediately by induction:
	$$
	\mathcal{L}_c^e([\alpha, f_1,...,f_n,f_{-n},...,f_{-1},\beta])
	\geq
	\mathcal{L}_c^e([\alpha, a_1,...,a_n,b_n,...,b_1,\beta]).
	$$ 
	Next, we try to reduce the loss by moving $\alpha$ to its ground truth position, which requires some detailed discussion on the rank position of $\alpha$.
	
	If $a_{i+1} \leq \alpha < a_{i}$ where we've implicitly defined $a_{n+1}=b_n$, then the difference between the current list and the ground truth is:  
	\begin{align}
		\Delta \mathcal{L}:&=\mathcal{L}_c^e([a_1,...,a_i,\alpha,a_{i+1},...,b_n,...,b_1,\beta]) - \mathcal{L}_c^e([\alpha, a_1,...,a_n,b_n,...,b_1,\beta])\notag \\
		&= \sum_{k=1}^{i}\ell(S_k^{\alpha}) - \ell(S_k), 
	\end{align}
	and 
	\begin{align}
		\ell(S_k^{\alpha}) - \ell(S_k) &= \log(\frac{\sum_{S_k^{\alpha}}e^{f_s}\sum_{S_k^{\alpha}}e^{-f_s} - 2(n-k+1)}
		{\sum_{S_k}e^{f_s}\sum_{S_k}e^{-f_s} - 2(n-k+1)})	\notag \\
		&= \log(1 - \frac{\sum_{f_s\in S_k-a_k}e^{\alpha-f_s}(e^{a_k-\alpha}-1)+e^{f_s-\alpha}(e^{\alpha-a_k}-1)}
		{\sum_{S_k}e^{f_s}\sum_{S_k}e^{-f_s} - 2(n-k+1)}) \notag \\
		&= \log(1 - \frac{\sum_{f_s\in S_k-a_k}e^{\alpha-f_s}(\delta_k-1)+e^{f_s-\alpha}(1/\delta_k-1)}
		{\sum_{S_k}e^{f_s}\sum_{S_k}e^{-f_s} - 2(n-k+1)}), \notag \\
	\label{dlk}
	\end{align}
	where $\delta_k:= e^{a_k-\alpha}$. Since
	$$ 
	e^{\alpha-f_s}\geq 
	\begin{cases}
		1 &  \text{on} \  f_s \leq \alpha, f_s \in S_k \\
		1/\delta_k & \text{on} \  f_s > \alpha, f_s \in S_k  
	\end{cases}
	$$
	we have
	\begin{align}
		&\sum_{f_s\in S_k-a_k}e^{\alpha-f_s}(\delta_k-1)+e^{f_s-\alpha}(1/\delta_k-1) \notag\\
		&\geq \#_{\{S_k-a_k,f_s\leq\alpha\}}(\delta_k -1+1/\delta_k-1) + \#_{\{S_k-a_k,f_s>\alpha\}}(1-1/\delta_k+1-\delta_k) \notag \\
		&= (\delta_k+1/\delta_k-2)(\#_{\{S_k-a_k,f_s\leq\alpha\}}- \#_{\{S_k-a_k,f_s>\alpha\}})\notag\\
		&\geq 0,
		\label{ieq_cnt}
	\end{align}
	where $\#_{\{S\}}$ denotes the number of elements in the set $S$. The last inequality comes from the fact that $b_k\leq...\leq b_n\leq a_n\leq...\alpha\leq...a_k$. Plugging the inequality \ref{ieq_cnt} into equation \ref{dlk}, we have
	\begin{equation}
		\ell(S_k^{\alpha}) - \ell(S_k) \leq  0. 
	\end{equation}
	
	Due to the symmetry of $\alpha$ and $\beta$, we can similarly prove that moving $\beta$ to its ground truth position will reduce the loss by taking the scores to be their opposite numbers.
\end{proof}

\begin{remark}
In Appendix \ref{appendixA}, we also prove the case $b_{n-1} \leq \alpha \leq b_n$ and discuss the case $\alpha \leq b_{n-1}$. Though Theorem \ref{thm2} states that, conditioned on the top half and the bottom half having been split correctly, minimizing $\mathcal{L}_c^e$ will recover the true permutation, this condition might not be necessary. In fact, we've done extensive numerical experiments to test whether $\mathcal{L}_c^e$ is not consistent with the 0-1 loss, yet no counterexamples have been found.
\end{remark}

From the proof, we can see that there are basically two ways to let $\mathcal{L}_c^e$ decrease in the permutation space:
\begin{itemize}
	\item putting higher scores at the top, lower at the bottom.
	\item putting small score differences (less distinguishable) pairs at the mid.
\end{itemize}
The second tendency, on the one hand, brings upon difficulty in proving the permutation level 0-1 consistency; on the other hand, it actually coincides with the stock market: the cross-sectional distribution of stock returns are approximately normal and those who lie in the middle have very little differences. From a long-short perspective, it is reasonable to put aside stock pairs that we have no opinion whether one will dominate the other in terms of return.

\section{Empirical studies}
\label{empirical}
In this section, we explore the empirical performance of ListFold in China A-share market.  Our data and code have been made open access at Github: \url{https://github.com/TCtobychen/ListFold}. We denote ListFold-sgm and ListFold-exp respectively for $\mathcal{L}_c^s$ and $\mathcal{L}_c^e$. We will compare our model with multilayer layer perception (MLP), ListMLE and Song's fitting two ListMLE (denoted as List2MLE). Since our work focuses on the loss function, we will adopt a same neural network structure and training method for all these algorithms. MLP stands as a value prediction representative and the others represent the rank wisdom. The loss function for MLP is as follows:
$$
\mathcal{L}_{MLP}(f,r) = \frac{1}{n}\sum_{i=1}^{n} (r_i-f_i)^2.
$$

For different algorithms we construct long-short portfolios accordingly. Then the evaluation metric is taken to be the portfolio performance measure. As far as rank prediction is concerned, a generalized NDCG is proposed for evaluation as well. This section is divided into three parts: a brief summary of the data and training, the network structure for the scoring function and the evaluation of the performance.

\subsection{Data and Training}
Our dataset consists of 631 weekly observation on 3712 stocks with 68 factors.\footnote{Data is obtained from the Wind database, \url{https://www.wind.com.cn/NewSite/data.html}}. The date is from 2006-12-29 to 2019-04-19. By setting a threshold that the percentage of missing values is less than 0.1\%, we filter out 80 stock. These 80 stocks are mostly listed in the HS300 index and they are highly liquid. The 68 factors are mainly some common factors and their names are listed below in Table \ref{factornames}.

\begin{table}
	\centering
	\caption{Factor names}
	\vspace{0.5ex}
	\label{factornames}
	\resizebox{\textwidth}{20mm}{
	\begin{tabular}{ccccccc}
	\hline
	\hline
	alpha\_100w & amount\_21 & amount\_5 & amount\_63 & amount\_div & avg\_volume\_21 & avg\_volume\_5 \\
	avg\_volume\_63 & beta\_100w & close\_low\_high & close\_s\_vwap5 & close\_vwap5 & c\_l2\_ibm & dlt\_miclo \\
	highlow\_1 & highlow\_12 & highlow\_3 & highlow\_6 & ibm\_close & ibm\_svlo & IR\_netasset\_252 \\
	IR\_roe\_252 & l2\_ibm\_ewma & l2\_lbm\_ewma & magm\_yop & ma\_crossover\_15\_36 & net\_assets& n\_buy\_value\_small\_order \\
	pb & pcf\_gm  & z\_sde\_pe & q\_s\_fa\_yoyocf & rank\_amount\_div & rank\_close\_low\_high & rt\_10 \\
	\hline
	rt\_126 & rt\_12\_1 & rt\_15 & rt\_21 & rt\_252 & rt\_5 & rt\_5\_Skewness\_10 \\
	rt\_5\_Skewness\_15 & rt\_5\_Skewness\_20 & rt\_5\_Skewness\_5 & rt\_63 & std\_deviation\_100w & yop\_pe & s\_dq\_mv \\
	yop\_pcf & s\_val\_mv & z\_rank\_pe & trk\_rk\_pe\_re & ttm\_pcf & ttm\_pe & ttm\_ps \\
	ttm\_roa & ttm\_roe & turnover\_21 & turnover\_5 & turnover\_63 & vol\_1 & vol\_12 \\
	vol\_3 & vol\_6 & yieldvol\_1m & yieldvol\_3m & yieldvol\_6m\\
	\hline
	\end{tabular}}
\end{table}

Since the dataset spreads over 13 years we train the model in a rolling basis. We use every 300 weeks as the training set and the next 16 weeks as the test set. We split the data into mini batches of size 32 and perform the min-max normalization every 300 weeks. In total, we got 320 weekly data as the test set that rolls from 2012-11-30 to 2019-02-01. During each training process, we let each algorithm view 1000 mini batches.

\subsection{Scoring Function}
We use a 4-layer fully-connected neural network to learn the scoring function $f$. Namely, the shape of the network structure is $[68\times 136 \times 272 \times 34\times 1]$. For each layer we embed a ReLU layer as the activation function \footnote{For MLP the last ReLU layer is not necessary since there are negative returns.}. The reason why we extend the dimension of features in the hidden layers is that we want the factors to interact with each other to generate more features. Note that this scoring function is the same for all documents, i.e. the parameters are the same for $f_1,...,f_n$.
 
\begin{figure}[H]
	\centering
	\includegraphics[scale=0.4]{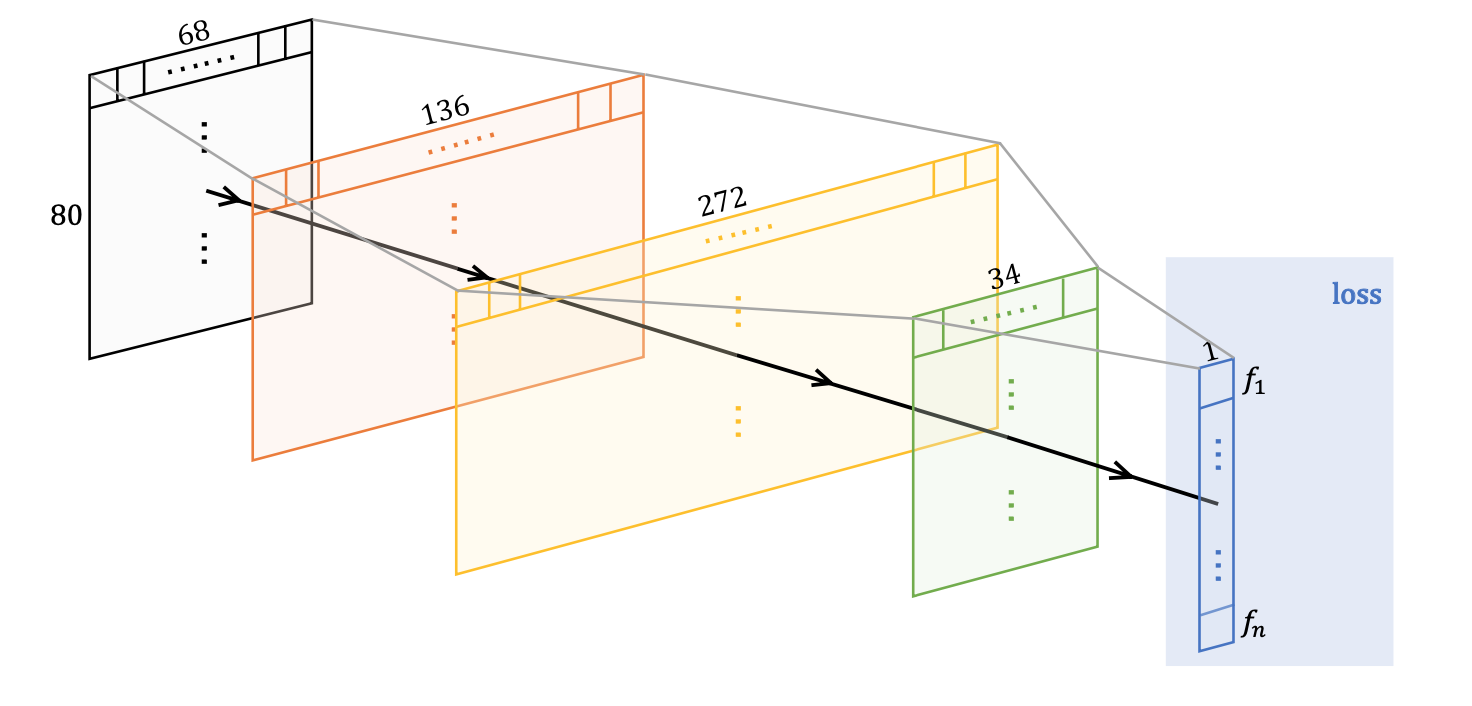}
	\label{mlp}
	\caption{The neural network to learn the scoring function}	
\end{figure}

We emphasize that we use no future information. We regard each week as independent sample and our network doesn't generate features along the time axis. This enables us to do the mini batch shuffling seamlessly in the training procedure. In practice, we can generalize the model to include the information along the time axis simply by feeding those information as time series factors. 

\subsection{Performance}
To evaluate our model, we first investigate the out-of-sample performance of our portfolio. Then we give evaluations from a ranking perspective in the hope that it can help us better decompose the profit and indicate our model's potential in IR criterions.

\subsubsection{Portfolio Metric}
Based on the obtained rank prediction, we build two kinds of strategies. One is to long the top 10\% stocks and short the bottom 10\%. The other is to long the top 10\% stocks and short the average of all the stocks. We allocate equal weight to each stock within the same direction. Note that for MLP the rank prediction will be given by sorting the predicted returns. The motivation for shorting the average is to approximate the long-short performance if we have to use the index future as a substitute of the short leg. Together we will have \textit{ListFold-exp, ListFold-sgm, ListMLE, List2MLE, MLP} and their version of shorting the average, denoted as: \textit{ListFold-exp-sa, ListFold-sgm-sa, ListMLE-sa, MLP-sa}. Note that for List2MLE, we've assigned what to short so we won't consider List2MLE shorting the average. At each week, we invest a fixed amount of nominal capital \$1. The out-of-sample pnl without transaction fee is plotted in Figure \ref{pnl-ls8}, with the average being a proxy of the baseline.
\begin{figure}
	\centering
	\includegraphics[width=\textwidth]{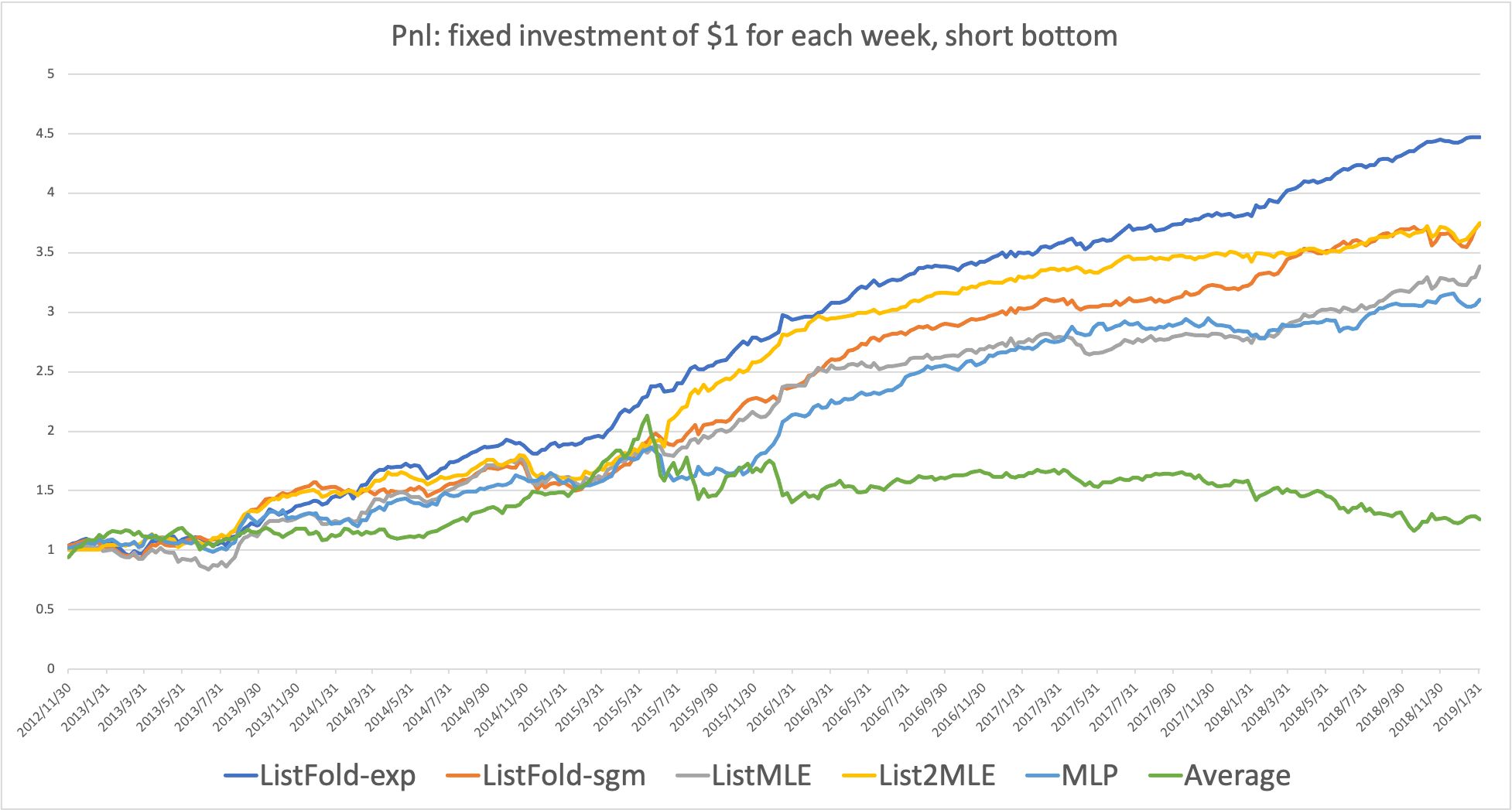}
	\caption{Pnl: fixed investment, long short 8 stocks}
	\label{pnl-ls8}
\end{figure}
\begin{figure}
	\centering
	\includegraphics[width=\textwidth]{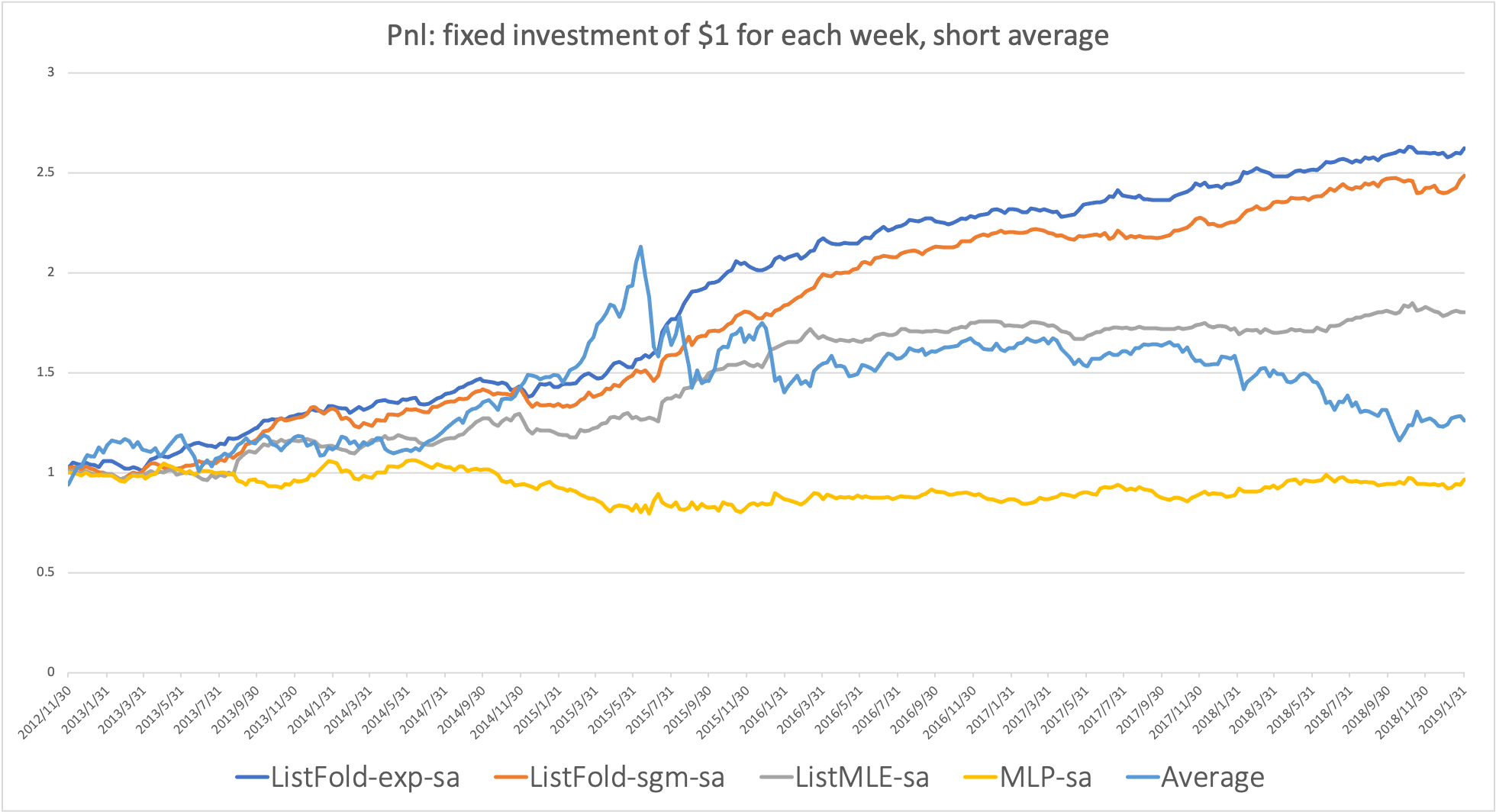}
	\caption{Pnl: fixed investment, long 8 short the average}
	\label{pnl-lsave}
\end{figure}

From Figure \ref{pnl-ls8}, we can see how the rank perspective contributes to the portfolio. Under the same neural network, same data and same training procedure, after 6 years the net value of MLP is only 2/3 of the net value of ListFold-exp. Also, ListFold-exp achieves the highest pnl. It is not so surprising that ListFold-exp outperforms ListFold-sgm given their consistency discussion, yet it is a surprise that ListFold-exp outperforms List2MLE. We take the difference of ListFold-exp and List2MLE for each week, and under normal assumption the t-statistics is 1.45. Looking into the detailed weekly positions, we find List2MLE indeed picks out some stocks to both long and short at the same time: the average number of the overlap is 0.42 stock per week. 

Comparing Figure \ref{pnl-ls8} and Figure \ref{pnl-lsave} we observe that ListFold-exp and ListFold-sgm has more advantage in the long leg. Almost all the excess return between listFold-exp and MLP comes from the long leg. Since 2016 ListMLE and MLP has ceased to be profitable if they short the average, whereas both ListFold methods still print money. One explanation could probably be that machine learning algorithms crowded into China A share since then.

Setting the annualized risk free rate $r_f$ at 3\%, the transaction cost per trade 30 bps in total (such as tax, spread crossing and getting short), we summarize the mean, standard deviation, Sharpe ratio and max drawdown in Table \ref{pnl-stat}. We also calculate the average trading turnover(TRV for short) which is the non-overlapped stock ratio for two consecutive weeks' positions. If TRV is less than 1, we can suffer less transaction fee. The first panel displays the strategies that short the bottom and the second panel corresponds to shorting the average.
\begin{table}
	\centering
	\caption{Portfolio statistics of the strategies}
	\vspace{0.5ex}
	\label{pnl-stat}
	\begin{tabular}{cccccc}
	\hline \hline
	 & ListFold-exp & ListFold-sgm & ListMLE & List2MLE & MLP\\
	$\mu-r_f$ & 0.38 & 0.26 & 0.20 & 0.26 & 0.16 \\
	$\sigma$ & 0.19 & 0.20 & 0.22 & 0.20 & 0.22 \\
	SR & 2.01 & 1.27 & 0.91 & 1.29 & 0.72 \\
	MDD & 0.14 & 0.25 & 0.23 & 0.21 & 0.28 \\
	TRV & 0.48 & 0.45 & 0.45 & 0.46 & 0.39 \\
	\hline
	 & ListFold-exp-sa & ListFold-sgm-sa & ListMLE-sa & List2MLE-sa & MLP-sa\\
	
	$\mu-r_f$ & 0.08 & 0.06 & -0.06 & $\times$ & -0.19 \\
	$\sigma$ & 0.11 & 0.11 & 0.10 & $\times$ & 0.11 \\
	SR & 0.71 & 0.50 & -0.53 & $\times$ & -1.79 \\
	MDD & 0.09 & 0.10 & 0.12 & $\times$ & 0.27 \\
	\hline \hline
	\end{tabular}
\end{table}

All these strategies have a low volatility due to their long-short nature. ListFold-exp outperforms List2MLE with a slightly higher turnover. For a larger stock pool and higher frequency data we expect the turnover would increase significantly.

\subsubsection{Rank Metric}

For the rank metric, we use Spearman's $\rho$, NDCG. We also propose $\text{NDCG}@\pm k$ as a generalization of NDCG that emphasizes both the top and the bottom:
\begin{definition}
$\text{NDCG}@\pm k$ is defined as the average of NDCG@k and the reverse labeled NDCG@-k:
	\begin{align}
	\text{NDCG}@\pm k(\pi,l) &= (\text{NDCG@k}(\pi,l)+\text{NDCG@-k}(\pi,\tilde{l}))/2 \notag \\
	&= \frac{1}{2Z_k}\left(\sum_{j=1}^{k}G(l_{\pi^{-1}(j)})\eta(j) + \sum_{j=1}^{k}G(\tilde{l}_{\tilde{\pi}^{-1}(j)})\eta(j)\right), \notag
\end{align}
\end{definition}
\noindent where $\tilde{l}, \tilde{\pi}$ are the reversed label and the revered list. They form the symmetric metric of NDCG@k at the bottom of a list. For instance, if $\pi = [a,b,c,d],l=[3,2,4,1]$ which implies that the true permutation $\pi_0$ is $[d,a,b,c]$, then $\tilde{\pi} = [d,c,b,a], \tilde{l} = [4,1,3,2]$. Note that the Spearman's $\rho$ is the Information Coefficient (IC) calculated in rank. To use NDCG type metrics, we also need to transform the returns into levels: we label the top 10\% as 10, then the top 10\% to 20\% as 9 etc until the bottom 10\% are labeled 1. We denote ListMLE-rvs for the ListMLE with reverse labelling. The statistics of rank metric are summarized in Table \ref{rank-stat}.
\begin{table}
	\centering
	\caption{Statistics of rank metric}
	\vspace{0.5ex}
	\label{rank-stat}
	\begin{tabular}{cccccc}
	\hline \hline
	   & ListFold-exp & ListFold-sgm & ListMLE & ListMLE-rvs & MLP\\
	IC & \boxed{0.079} & 0.077 & 0.077 & 0.057 & 0.055\\
	NDCG & 0.613 & 0.611 & 0.630 & \boxed{0.632} & 0.628 \\
	NDCG@8 & 0.229 & 0.234 & 0.183 & \boxed{0.244} & 0.201\\
	NDCG@-8 & 0.301 & 0.286 & \boxed{0.324} & 0.280 & 0.238\\
	NDCG@$\pm$8 & \boxed{0.265} & 0.260 & 0.254 & 0.262 & 0.219\\
	\hline
	\end{tabular}
\end{table}
Although LisMLE has a high NDCG, as many previous study has suggested, its NDCG@8 is surprisingly not as good on our dataset. List2MLE take the long leg of ListMLE and the short leg of ListMLE-rvs so its NDCG$\pm$8 is actually the lowest. To some extent, this undermines the credit of List2MLE. Moreover, IC coincides with the pnl performance fairly well which gives endorsement to IC's popularity among industries.

 
\subsection{Robustness}
In this part, we first examine the robustness of the position cutoff parameter $k$, i.e. the top-$k$ and bottom-$k$ we decide to long and short in our portfolio. For different models, we plot out a heatmap for different number of stocks to long-short. The columns are the models and the row stands for the position cutoff parameter $k$. For example, for the 8th row we long the top 8 stocks and short the bottom 8 stocks. The number in each cell is the average weekly return (in bps) we achieve out of sample, without transaction fee. The larger this value, the warmer the color.

\begin{figure}
	\centering
	\includegraphics[width=0.4\textwidth]{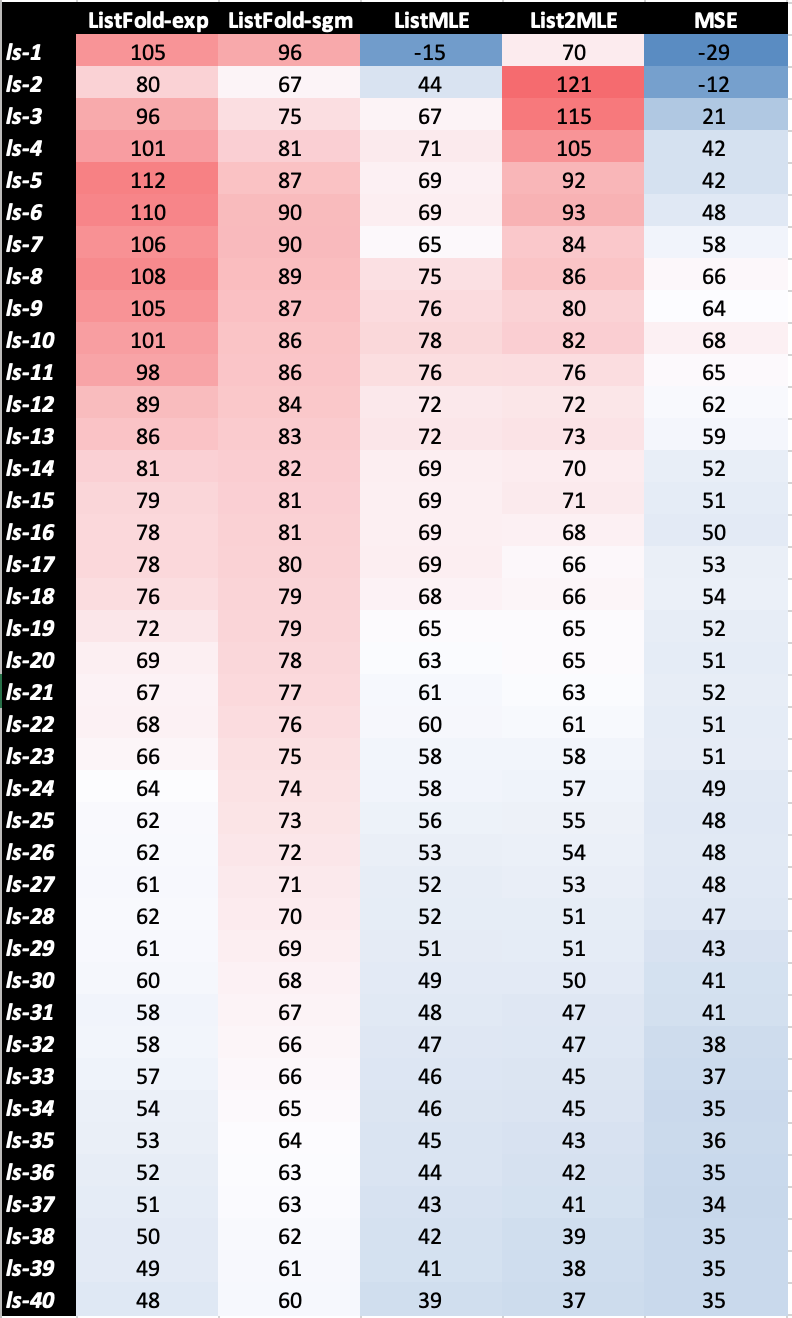}
	\caption{Heatmap for long-short top-k pairs}
	\label{heatmap}
\end{figure}

If our prediction matches the ground truth perfectly, we would expect the pnl of the portfolio  gradually decreases towards 0 as $k$ goes from 1 to 80. Figure \ref{heatmap} demonstrates vividly the advantage of ListFold-exp. All the four rank methods are robust regarding $k$. If the value from LS-1 to LS-8 is approximately decreasing, we can also promote the pnl by assigning larger weight to the precedent pairs. In addition to that, we also notice ListFold-sgm has the best long-short-40 performance among all five models, and it's performance is stable. This coincides with the fact that ListFold-sgm is consistent with the binary classification loss. So if one's task is to long 50\% and short 50\%, then ListFold-sgm should be a good choice.

Next, we examine the robustness of the mini batch sizes for the ranking methods. We train the model under different mini batch sizes and keep the model always view 1000 batches. The average weekly return in bps is summarized in Table \ref{mini batch size}. Overall a small mini batch size is not suitable for ListFold-exp and ListFold-sgm due to the non-convexity. Although a more complicated loss function generally requires more data, early stopping rule is also necessary.
\begin{table}
	\centering
	\caption{Average weekly return of long-short 8 stocks (in bps)}
	\vspace{0.5ex}
	\label{mini batch size}
	\begin{tabular}{cccccc}
	\hline \hline
	mini batch size & ListFold-exp & ListFold-sgm & ListMLE & ListMLE-rvs & List2MLE\\
	8 & 57.1 & 55.9 & \boxed{85.6} & 78.3 & 78.4\\
	16 & 72.6 & 79.9 & \boxed{82.3} & 71.0 & 73.5\\
	32 & \boxed{108.1} & 88.7 & 74.5 & 90.8 & 85.9\\
	64 & \boxed{99.3} & 76.2 & 88.1 & 70.3 & 78.3\\
	128 & \boxed{84.0} & 82.0  & 76.3 & 77.6 & 77.4\\
	\hline
	\end{tabular}
\end{table}

Finally we would like to mention that our strategy is not restricted to a weekly manner nor 80 stocks. It can be seamlessly transferred to the practitioner's own factor dataset. The high Sharpe ratio of our model also encourages leverage in real production. Since the China A-share market goes through an extraordinary bull-bear circle in our testing period, our strategy has demonstrated its robustness against the market turbulence.

\section{Conclusion}
\label{conclusion}
We've suggested a new perspective of combining factors for constructing long-short portfolio. Following the learn-to-rank method we've proposed a new type of loss function that aims to select long-short pairs listwisely. It is shift-invariant and probabilistically explained. For different transformation functions it can consist with binary classification loss or permutation level 0-1 loss. Our model can be viewed as a supplementary tool for investigating non-order sensitive loss and it may also inspire frameworks that unify pairwise and listwise surrogate loss functions.

Our empirical study in China A-share market has achieved 39\% annualized return with a Sharpe ratio of 2.07 for 6 years. Not only does it demonstrate the advantage of rank prediction over value prediction, but also it depicts the power of our loss function, especially the ListFold-exp. We've done a thorough evaluation of the different models from both a financial and ranking perspective. It turns out that the loss functions we proposed have significant advantage over the others. A byproduct of our research is that we empirically verify IC is better than NDCG type ranking metrics for evaluating alpha strategies.

For future study, regarding the theoretical analysis, it is worthwhile to further investigate the consistency of ListFold-exp. More generally, to give a characterization of different transformation functions that lead to different true loss functions. From a practitioner's view, it also worths to try ListFold on their own factor dataset or neural network. As a matter of fact, in view of getting short is usually more expensive, one may want to be more correct in average on the bottom ranks. So it might also be interesting to combine ListFold with other loss functions or add neural network components (like auxiliary task) to the current architecture.
\newpage
\bibliographystyle{rQUF}
\bibliography{Reference_QF0317}

\newpage
\appendix
\renewcommand{\appendixname}{Appendix~\Alph{section}}
\section{More on the consistency of $\mathcal{L}_c^e$}
\label{appendixA}
We want to get rid of the assumptions of Theorem \ref{thm2} and consider $\bold{f}$ as a permutation on all $a_i's$ and $b_i's$. So we continue with more discussions on $\alpha$'s ground truth position. 
\vspace{2ex}
  
$\diamond$ CASE \ 2

If $b_{n-1} \leq \alpha < b_n$, similar with the argument above, we took the difference 
\begin{align}
	\Delta \mathcal{L}:&=\mathcal{L}([a_1,...,a_n, b_n,\alpha,b_{n-1},...,b_1,\beta]) - \mathcal{L}([\alpha, a_1,...,a_n, b_n,...,b_1,\beta])\notag \\
	&= 2(\alpha- b_n) + \sum_{k=1}^{n}\ell(S_k^{\alpha}) - \ell(S_k). \label{ieq5} 
\end{align}
and denote $c_k=2(n+1-k)$ we have
\begin{align}
	\ell(S_k^{\alpha}) &- \ell(S_k) = \log \left( 
	1 + \frac{\delta_k-1}{\delta_k} \cdot \frac{\delta_k \sum_{S_k}e^{f_s-a_k} - \sum_{S_k}e^{a_k-f_s} - (\delta_k-1)}{\sum_{S_k}e^{f_s}\sum_{S_k}e^{-f_s}-c_k} \right)\notag \\
	&= \log \left( 
	1 + \frac{\delta_k-1}{\delta_k} \cdot \left[ \frac{\delta_k \sum_{S_k-b_n}e^{f_s-a_k} - \sum_{S_k-b_n}e^{a_k-f_s} - (\delta_k-1)}{\sum_{S_k}e^{f_s}\sum_{S_k}e^{-f_s}-c_k} +
	\frac{\delta_ke^{b_n-a_k}-e^{a_k-b_n}}{\sum_{S_k}e^{f_s}\sum_{S_k}e^{-f_s}-c_k}\right ]
	\right) \notag \\
	& \leq \log \left( 
	1 + \frac{(\delta_k-1)^2}{\delta_k} \cdot \frac{\#_{\{S_k-bn,f_s>\alpha\}}- \#_{\{S_k-b_n,f_s\leq\alpha\}}-1}{\sum_{S_k}e^{f_s}\sum_{S_k}e^{-f_s}-c_k} + 
	\frac{e^{b_n-\alpha}+e^{\alpha-b_n}-e^{a_k-b_n}-e^{b_n-a_k}}{\sum_{S_k}e^{f_s}\sum_{S_k}e^{-f_s}-c_k}
	\right)\notag \\
	& \leq \log\left( 1 + \frac{e^{b_n-\alpha}+e^{\alpha-b_n}-e^{a_k-b_n}-e^{b_n-a_k}}{\sum_{S_k}e^{f_s}\sum_{S_k}e^{-f_s}-c_k}
	\right).
\end{align}
The last inequality follows from the fact that 
$$
\#_{\{S_k-bn,f_s>\alpha\}}- \#_{\{S_k-b_n,f_s\leq\alpha\}}-1 = 0, \quad \text{for} \ b_{n-1}\leq \alpha < b_n.
$$	
Also note that in the denominator $\sum_{S_k}e^{f_s}\sum_{S_k}e^{-f_s}=\sum_{f_s,f_t \in S_k}e^{f_s-f_t}$ is a function of the pairwise differences, and the difference between the elements in $\{S_k,f_s>\alpha\}$ and $\{S_k,f_s\leq\alpha\}$ is at least $b_n -\alpha$, so 
$$
\sum_{S_k}e^{f_s}\sum_{S_k}e^{-f_s} - c_k \geq \#_{\{S_k, f_s>\alpha\}}\cdot\#_{\{S_k, f_s\leq\alpha\}}\cdot(e^{b_n-\alpha}+e^{\alpha-b_n})
$$
Therefore, using ${e^{a_k-b_n}+e^{b_n-a_k}} \geq 2$, we have
\begin{align}
	\sum_{k=1}^{n}\ell(S_k^{\alpha}) - \ell(S_k) &\leq 	\sum_{k=1}^{n}\log\left( 1 + \frac{e^{b_n-\alpha}+e^{\alpha-b_n}-e^{a_k-b_n}-e^{b_n-a_k}}{\sum_{S_k}e^{f_s}\sum_{S_k}e^{-f_s} - c_k}
	\right)\notag \\
	& \leq \sum_{k=1}^{n-1} \log\left(1+\frac{e^{b_n-\alpha}+e^{\alpha-b_n}-2}{k(k+1)(e^{b_n-\alpha}+e^{\alpha-b_n})}\right) + \log \left(\frac{e^{b_n-\alpha}+e^{\alpha-b_n}}{e^{a_n-b_n}+e^{b_n-a_n}}\right) \notag \\
	& \leq (1-\frac{2}{{e^{b_n-\alpha}+e^{\alpha-b_n}}})(\frac{1}{1\cdot2}+\frac{1}{2\cdot3}+...) + \log (\frac{e^{b_n-\alpha}+e^{\alpha-b_n}}{2}) \notag \\
	& \leq (b_n -\alpha) + (b_n-\alpha) = 2(b_n-\alpha).
\end{align}
The last inequality follows from the inequality
$$
x + \frac{2}{e^x+e^{-x}} \geq 1, \text{for} \ x \geq 0,
$$
which can be easily proven by taking the derivative.
Hence, we have proved that equation (\ref{ieq5}) is less equal to zero, so the result stays true in this case. 

\vspace{2ex}

$\diamond$ CASE \ 3

If $b_{i} \leq \alpha < b_{i+1} \ \text{for some } 1\leq i \leq n-2$, similar with the argument above, we took the difference 
\begin{align}
	\Delta \mathcal{L}:&=\mathcal{L}([a_1,...,a_n, b_n,...,b_{i+1},\alpha,b_i,...,b_1,\beta]) - \mathcal{L}([\alpha, a_1,...,a_n, b_n,...,b_1,\beta])\notag \\
	&= 2(\alpha- b_n) + \sum_{k=1}^{i+1} \left( \ell(S_k^{\alpha}) - \ell(S_k) \right) + \sum_{k=i+2}^{n}\left(\ell(S_k^{b_{k-1}}) - \ell(S_k) \right ). \label{ieq8} 
\end{align}
For items of the form $\ell(S_k^{\alpha}) - \ell(S_k)$, by simple counting $\#_{\{a_k,...,a_{i+1}\}}=\#_{\{b_k,..,b_{i}\}}+1$, and the denotation that $M_i = \{b_{i+1},...,a_{i+2}\}, \gamma = a_{i+1}$, we got 
\begin{align}
	\ell(S_k^{\alpha}) - \ell(S_k) 
	& \leq \log \left (1 + \frac{\sum_{f_s\in M_i} \cosh (\alpha -f_s) -  \cosh (a_k -f_s)}{\sum_{f_s,f_t\in S_k, s\neq t} \cosh (f_s-f_t) }\right) \notag \\
	& \leq \log \left (1 + \frac{\sum_{f_s\in M_i} \cosh (\alpha -f_s) -  \cosh (\gamma -f_s)}{\sum_{f_s,f_t\in S_k, s\neq t} \cosh (f_s-f_t) }\right) \notag \\
	& \leq \log \left (1 + \frac{\sum_{f_s\in M_i} \cosh (\alpha -f_s) -  \cosh (\gamma -f_s)}{\sum_{f_s,f_t\in \{M_i+k\alpha+(k+1)\gamma\},s\neq t} \cosh (f_s-f_t) }\right) \notag \\
	& = \ell( \{M_i+(k+1)\alpha+k\gamma\}) - \ell( \{M_i+k\alpha+(k+1)\gamma\})
\end{align}
Thus we have
\begin{align}
	\Delta \mathcal{L}&= 2(\alpha- b_n) + \sum_{k=1}^{i+1}\ell(S_k^{\alpha}) - \ell(S_k) + \sum_{k=i+2}^{n}\ell(S_k^{b_{k-1}}) - \ell(S_k) \notag \\
	&\leq 2(\alpha- b_n) + \sum_{k=1}^{i+1}{\ell( \{M_i+(k+1)\alpha+k\gamma\}) - \ell( \{M_i+k\alpha+(k+1)\gamma\})}\notag \\ & \quad + \sum_{k=i+2}^{n}\ell(S_k^{b_{k-1}}) - \ell(S_k) \notag \\
	&= \mathcal{L}([\gamma,...,\gamma,M_i,\alpha,...,\alpha,\beta])- \mathcal{L}([\alpha,\gamma,...,\gamma,M_i,\alpha,...,\alpha,\beta])
\end{align}
which is a special case that $\alpha$ is the minimal number. Therefore, all we need to prove is the case that $\alpha$ is the smallest number.

Although the rest is intuitively correct that we wont't put a small number at the beginning, it is actually hard to prove analytically. We've done extensive numerical experiments and the simulation finds no counterexamples. We leave this as a future study. 

\end{document}